\documentclass[11pt,espf]{article}

\usepackage{graphicx} 
\usepackage{epsfig,color} 
\usepackage{amsmath} 
\usepackage{amssymb,mathrsfs}  
\usepackage{theorem}
\newtheorem{definition}{Definition}

\newenvironment{proof}{\noindent{\bf Proof:}\rm}{\hfill\hfill$\Box$\par\medbreak}

\newtheorem{theorem}{Theorem}[section]
\newtheorem{proposition}[theorem]{Proposition}

\newtheorem{lemma}[theorem]{Lemma}

{\theorembodyfont{\rmfamily} 
\newtheorem{remark}[theorem]{Remark}

}
\usepackage[vlined,ruled,linesnumbered]{algorithm2e} 
\SetKwInput{KwAssumes}{Assumes}

\newcommand{\e}{\mathbf{e}}

\newcommand{\PP}{\mathbf{P}}
\newcommand{\p}{\mathbf{p}}
\newcommand{\C}{C}
\newcommand{\q}{\mathbf{q}}
\newcommand{\x}{\mathbf{x}}

\newcommand{\Hull}{\operatorname{Hull}}
\newcommand{\Cone}{\operatorname{Cone}}

\newcommand{\argmin}{\operatorname{argmin}}
\newcommand{\kmax}{\operatorname{max}_k}
\newcommand{\subscr}[2]{{#1}_{\textup{#2}}}
\newcommand{\norm}[1]{\|#1\|}
\newcommand{\abs}[1]{|#1|}

\newcommand{\intHull}[1]{{\Hull_{#1}^o}}
\long\def\cut#1{{}}

\begin{document}

\title{\LARGE \bf
$k$-Capture in Multiagent Pursuit Evasion, \\
or the Lion and the Hyenas} 

\author{Shaunak D. Bopardikar\thanks{Shaunak D. Bopardikar is with
    United Technologies Research Center Inc., Berkeley, CA and Subhash
    Suri is with the Department of Computer Science, University of
    California Santa Barbara, CA. This work was performed when the
    first author was with the Department of Electrical and Computer
    Engineering, University of California Santa Barbara, CA,
    USA. Emails: \texttt{bshaunak@gmail.com, suri@cs.ucsb.edu}}
  \qquad Subhash Suri}


\maketitle

\begin{abstract}
  We consider the following generalization of the classical pursuit-evasion problem, 
  which we call \emph{$k$-capture}. A group of $n$ pursuers (hyenas) wish to capture 
  an evader (lion) who is free to move in an $m$-dimensional Euclidean space, the pursuers 
  and the evader can move with the same maximum speed, and at least $k$ pursuers must
  \emph{simultaneously} reach the evader's location to capture it.  If fewer than $k$ 
  pursuers reach the evader, then those pursuers get destroyed by the evader. Under what 
  conditions can the evader be $k$-captured? We study this problem in the discrete time,
  continuous space model and prove that $k$-capture is possible if and only there exists 
  a time when the evader lies in the interior of the pursuers' $k$-Hull.   When the pursuit 
  occurs inside a compact, convex subset of the Euclidean space, we show through an 
  easy constructive strategy that $k$-capture is always possible.
\end{abstract}

\section{Introduction}

We consider a variant of the pursuit-evasion game in which multiple
pursuers must simultaneously reach the evader's location to capture
it.  Specifically, an evader $\e$, who is free to move in an
$m$-dimensional Euclidean space, is being pursued by $n$ agents $\p_1,
\ldots, \p_n$.  The evader and the pursuers have identical motion
capabilities and, in particular, have equal maximum speed. Unlike the
classical pursuit evasion, our game requires at least $k$ pursuers to
\emph{simultaneously} reach the evader's location to capture it, for
some given value of $k \leq n$.  If fewer than $k$ pursuers attack
(reach) the evader, then those pursuers are destroyed by
the evader. We assume that no two players ever occupy the same
position in the environment \emph{except} at the moment of capture;
that is, co-location either ends the game or only one player survives
among the co-located ones. By disallowing co-location, we are assuming
a weaker model of pursuers, which may also be more realistic because
in many physical systems only one agent can occupy a point in the
space.
We call this version the \emph{$k$-capture pursuit evasion}, and
investigate necessary and sufficient conditions, as well as worst-case 
time bounds, for the $k$-capture.

\medskip

Pursuit-evasion games provide an elegant setting to study algorithmic and 
strategic questions of exploration or monitoring by autonomous agents. 
Their rich mathematical history can be traced back to at least 1930s 
when Rado posed the now-classical Lion-and-Man'' problem~\cite{JEL:86}: 
\emph{a lion and a man in a closed arena have equal maximum speeds; 
what tactics should the lion employ to be sure of his meal?}
The problem was settled by Besicovitch who showed that the man
can escape regardless of the lion's strategy~\cite{JEL:86}. An 
important aspect of this pursuit-evasion problem, and its solution, is 
the assumption of \emph{continuous time}: each player's motion is a 
continuous function of time, which allows the lion to get arbitrarily 
close to the man but never capture him. If, however, the players move in 
discrete time steps, taking alternating turns but still in 
continuous space, the outcome is different, as first conjectured by 
Gale~\cite{RKG:91} and proved by Sgall~\cite{JS:01}. 

\medskip

The distinction between continuous and discrete time models is significant 
albeit subtle. Formulations based on the continuous time lead to differential 
games, whose solution requires solving the Hamilton-Jacobi-Bellman-Isaacs 
(HJBI) equation. This is a partial differential equation, whose solution 
becomes intractable in complex scenarios. (See the seminal work of
Isaacs~\cite{RI:65} on several continuous time classical games
including the Homicidal Chauffeur game and the Game of Two Cars.)
Besides this \emph{theoretical} difficulty, one also faces the \emph{practical}
problem that continuous time solutions usually are expressed as a feedback 
law requiring an \emph{instantaneous} measurement of each player's 
position and its communication to the opponent. This is impractical from an 
implementation point of view, and especially problematic for non-smooth 
motions.

\medskip 

Consequently, discrete time alternate moves versions of pursuit evasion 
have been favored in recent past, especially due to their algorithmic
tractability. In these formulations, the evader and the pursuers move
in alternating time instants, with the evader moving first.  We note that 
a capture in this formulation is equivalent to the evader being inside a 
specified small neighborhood of the pursuer in the continuous time 
formulation. In the discrete time model, Sgall~\cite{JS:01} is able 
to circumvent the problem of lion approaching but not reaching the 
man in the continuous formulation, and shows that the lion can always 
capture the man in finite time inside a semi-open bounded environment.

\medskip

When the evader is free to move inside an \emph{unbounded}
environment, multiple pursuers are clearly required to keep the evader
from escaping.  The capture condition is the same as before: if at
some time $t$, \emph{any} pursuer can reach the position of the evader
then the latter is captured.  In this setting, it is known that the
evader can be captured if and only if it lies in the convex hull of
the pursuers~\cite{SK-CVR:05}.  Many other pursuit evasion problems
have also been studied, with focus on different types of
environments~\cite{LA-ASG-EMR:92,VI-SK-SK:05}, characterization of
environments in which a certain capture strategy
works~\cite{SBA-RB-RG:09}, visibility-based
pursuit-evasion~\cite{LJG-JCL-SML-DL-RM:98}, sensing
limitations~\cite{SDB-FB-JPH:07o,NK-VI:08} etc.  Finally, if both time
\emph{and} space are assumed to be discrete, then the underlying space
is represented as a graph with nodes and edges, and on each move a
player can move from one node to another by traversing the edge(s)
connecting them. The techniques in this formulation tend to be
different, and we refer the reader to a representative set of
papers~\cite{TDP:78, MA-MF:84, VI-NK:08, KK-SS:10}.

\medskip 

Our objective in this paper is to study the $k$-capture problem in the 
unbounded continuous space and discrete time framework.  In particular, we 
assume that a group of $n$ pursuers (hyenas) wish to capture an evader (lion) who is free 
to move in the $m$-dimensional Euclidean space.  The players take turns: the evader 
moves first, the pursuers move next and all of the pursuers can move 
simultaneously.  On its turn, each player can move anywhere inside a 
unit disk centered at its current position. (In other words, the
maximum speed of the players is normalized to one.)
We assume that no two players may occupy the same position in the
environment \emph{except} 
at the moment of capture. Technically, this assumption is used only to rule 
out the possibility of a trivial pursuer strategy where they partition themselves 
into size $k$ subgroups, with each subgroup moving as a ``meta pursuer.''
Co-location may also be unrealistic in many physical systems, and by disallowing 
it we only strengthen our results because pursuers without co-location are 
weaker in power than those with co-location.

\medskip

We say that the evader is $k$-captured, for some specified value of $k$, 
if after a finite time, at least $k$ pursuers reach the evader's location.
However, if fewer than $k$ pursuers reach the evader's location, then
the evader is able to capture (or destroy) those pursuers. In other
words, if at the end of a pursuer move, the evader occupies the same position as some of
the pursuers, then the game either ends ($k$-capture occurs), or all
the $j$, where $j < k$, pursuers at that location are captured leaving
only the evader.  We study the necessary and sufficient conditions
under which such a $k$-capture is possible, and derive bounds on the
worst-case time needed to achieve this. Additionally, we address a
version of this problem played in a compact and convex subset of a Euclidean space.

\medskip

In particular, our paper makes four main contributions. First, we show
that a necessary condition for $k$-capture is that the evader
must be located inside the $k$-Hull of the pursuers at the beginning
of every evader move. The $k$-Hull is the set of all points $p$ such
that any line through $p$ divides the given points into two sets of at
least $k$ points each. Second, we show that this simple $k$-Hull
condition is also sufficient. In other words, if there is ever a time
when this condition is satisfied, and in particular if it holds at
time $t=0$, then the pursuers can $k$-capture the evader in finite
time. Our proof of sufficiency is constructive, and based on a new
multi-pursuer strategy.  Third, we derive an upper bound for the time
needed to capture the evader, as a function of the initial positions
of the pursuers and of the evader. Finally, for a version of this
problem played in a compact and convex environment in a Euclidean
space, we design a novel strategy and show that the evader is
$k$-captured using $k$ pursuers.

\medskip

This paper is organized as follows. The problem formulation and the
necessity of the $k$-Hull condition for capture are presented in
Section~\ref{sec:problem}. Our multi-pursuer capture strategy and the
sufficiency of the $k$-Hull condition is presented in
Section~\ref{sec:suff}. A version of this problem played in a compact
and convex environment is analyzed in Section~\ref{sec:compact}. The
conclusions and future directions for this work are summarized in
Section~\ref{sec:conclusions}.

\section{Problem Formulation and Necessary Condition for $k$-Capture}
\label{sec:problem}

\medskip

Our pursuit-evasion game is played in an $m$-dimensional Euclidean space, with
$n$ pursuers $\p_1, \p_2, \dots,\p_n$ and a single evader $\e$. The
positions of these agents at any time $t$ are denoted as $\p_j (t)$,
for $j=1, 2, \ldots, n$, and $\e(t)$, where $t \in \mathbb{Z}_{\geq 0}$.  
In Section~\ref{sec:compact}, we also consider the capture problem in a compact
convex environment. 

\medskip

We assume that the game is played in discrete time using alternate moves:
on a turn, the evader moves first, all the pursuers simultaneously move next.
We assume a normalized maximum speed of one, meaning that each player can move to 
any position inside a closed ball of radius one centered at the player's current position. 
More precisely, the players' motions are described by the following equations:
\begin{align*}
\e(t+1) &= \e(t) + \mathbf{u}_e(t,\p_1(t),\dots,\p_n(t)), \\
\p_j(t+1) &= \p_j(t) + \mathbf{u}_{\p_j}(t,\e(t),\e(t+1)),
\end{align*}
where $\mathbf{u}_e$ and $\mathbf{u}_{\p_j}$ are unit vectors, termed as 
\emph{strategies} of the evader $\e$ and the pursuer $\p_j$, respectively. 
These motion equations say that each agent's strategy depends on the current positions
of all other players, and that each agent can move to any position
within distance one of its current position. (The apparent asymmetry 
in the equations of the evader and the pursuers is due to the fact that the evader moves
first, so the pursuers' moves can depend on the evader's positions at times $t$ \emph{and} 
$t+1$.)

\medskip

The capture occurs when evader is at the same location as some of the pursuers.
The $k$-capture of the evader requires at least $k$ pursuers, while fewer than 
$k$ pursuers are themselves captured by the evader.\footnote{%
We remark, however, that in 
the discrete time alternate moves model, the evader cannot force a pursuer's capture
because the pursuers move \emph{after} the evader. Indeed, if the evader
moves to the current location of a pursuer $\p$, then $\p$ can always move away 
from the evader at its turn. However, one cannot rule out a pursuers' strategy
that involves \emph{sacrificing} some of them to ultimately achieve $k$-capture.
}
Formally, we say that the evader is \emph{$k$-captured} if there exists a finite 
time $T$ and a subset $\C \subset \{\p_1, \ldots, \p_n \}$ of $k$ pursuers such that 
$\norm{\p_j (T)-\e(T)}\;=\;0$ but $\norm{\p_j (t)-\e(t)}>0$, for all $t<T$
and all $\p_j \in \C$.
In other words, the $k$-capture occurs at a time $T$ if at least $k$ pursuers 
simultaneously reach the evader's location at time $T$, and none of these pursuers
have ever been captured in the past.\footnote{%
While it is sufficient to ensure the safety of only the $k$
pursuers who perform the $k$-capture, in our strategy, \emph{all}
the pursuers will remain safe.} 
We say that the evader \emph{escapes} if there exists no finite time at 
which the pursuers $k$-capture the evader.
Finally, we require that no two players occupy the 
same point in the environment \emph{except} at the time of capture.

\medskip

We now formulate a necessary condition for $k$-capture, which is then
complemented by Section~\ref{sec:suff} that shows that this condition
is also sufficient.  Our necessary condition prescribes the location
of the evader relative to the locations of the pursuers for the
$k$-capture to occur. This condition is independent of the pursuers'
strategy: that is, if the condition is violated, then there always
exists an evader strategy for escape regardless of the pursuers'
strategy.

\medskip

Naturally, the convex hull of the pursuers' locations plays a key role in the game.
This is not surprising because the convex hull is precisely the set of all evader 
locations that are capturable in the classical single pursuer game, as is
well-known~\cite{SK-CVR:05}. 

\begin{lemma}
\label{lem:conv}
If the evader's initial location is not inside the interior of the convex hull
of the pursuers, then it cannot be $k$-captured, even for $k=1$.
\end{lemma}
\begin{proof}
  If the evader is not in the interior of the convex hull, then there
  exists a hyperplane through the evader's location such that all the
  pursuers lie in one (closed) half-space defined by the
  hyperplane. The evader simply escapes by moving perpendicular to
  this hyperplane, away from the pursuers, at maximum speed.
\end{proof}

\medskip

\subsection{The {\Large $k$}-Hull}

\medskip

When $k > 1$, we need a generalized notion of the convex hull. The standard
convex hull can be defined as the set of points so that any hyperplane tangent to
the hull contains at least one point of the hull in each of the two half-spaces.
If we require that at least $k$ points lie in each half-space, then we
get a structure called $k$-Hull, introduced by Cole, Sharir and 
Yap~\cite{RC-MS-CKY:87}, which also has intimate connections to other
fundamental structures in computational geometry such as 
$k$-levels and $k$-sets~\cite{HE:87}.

\begin{definition}[$k$-Hull]\label{def:khull}
  Let $S$ be a set of $n$ points in the plane, and let $k$ be an
  integer.  The $k$-Hull of $S$ denoted by $\Hull_k(S)$ is the set of
  points $p$ such that, for any hyperplane $\ell(p)$ through $p$, there are
  at least $k$ points of $S$ in each closed half-space of $\ell(p)$.
\end{definition}

Clearly, the standard convex hull is the same as the $1$-Hull, and it
is also easy to see that the $(k+1)$-Hull is contained in the
$k$-Hull. One can also show, using Helly's Theorem~\cite{EH:23}, that the
$k$-Hull is always non-empty for $k \leq \lceil n/(m+1)\rceil$, where
$m$ is the dimension of the underlying Euclidean space. We,
therefore, assume throughout this paper that $1 \leq k \leq \lceil n/(m+1)
\rceil$.  In particular, the standard convex hull in two dimensions, is well-defined for
3 or more non-collinear points, but $2$-Hull requires at least $n=5$
points in the plane. Fig.~\ref{fig:constraint} shows the two possible
configurations for $n=5, k=2$ for a planar environment.

\begin{figure}[htbp]
\centering
\includegraphics[width=0.8\columnwidth]{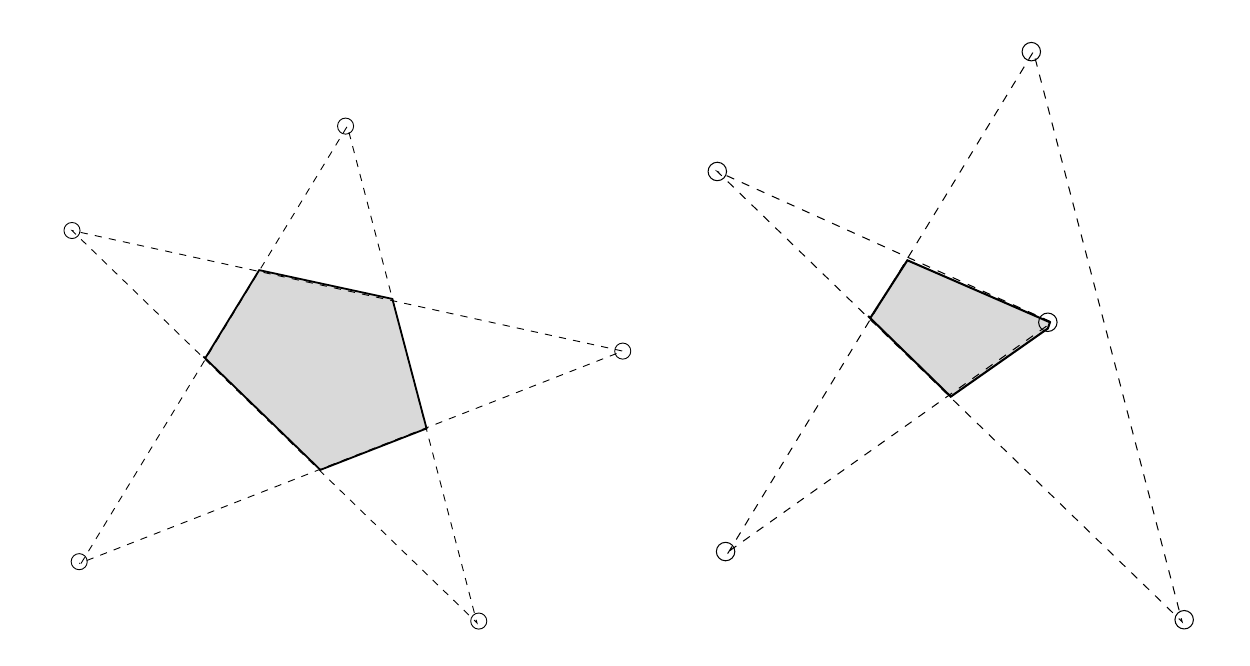}
\caption{Illustrating the $k$-Hull. Configurations for $2$-Hull of
  $n=5$ points in the plane.}
\label{fig:constraint}
\end{figure}

\begin{remark}[$k$-Hull Computation] 
  While computational complexity is the not focus of our paper, we do
  point out that $k$-Hulls are also efficiently computable.  Under the
  point-hyperplane duality, they correspond to the level $k$ in an
  arrangement of hyperplanes, and therefore computed easily in
  $O(n^m)$ time in $m$ dimensions. The bound can be improved somewhat
  using more sophisticated algorithms and analysis. For instance, in
  the two-dimensional plane, $k$-Hull contains at most $O(n k^{1/3})$
  vertices, and using dynamic convex hull data structures, it can be
  computed in worst-case time $O(n k^{1/3} \log^2
  n)$~\cite{RC-MS-CKY:87}.
\end{remark}

\medskip

\subsection{The Necessary Condition}

\medskip

In a pleasing generalization of 1-capture, it turns out that the $k$-Hull of 
the pursuers' locations is precisely the set of evader locations that are 
$k$-capturable.  Throughout, we will use the notation
$\intHull{k}$ to denote the \emph{interior} of the $k$-Hull.
We have the following theorem stating our necessary condition.

\begin{theorem} \label{thm:necessary} The evader $\e$ can be
  $k$-captured by pursuers $\p_1, \p_2, \ldots, \p_n$ only if
                $$\e(t) \:\in\: \intHull{k} (\p_1(t),\dots,\p_n(t)),$$
	at all times $t\in \mathbb{Z}_{\geq 0}$, where $\intHull{k}$ is the interior of 
	the $k$-Hull.
\end{theorem}
\begin{proof}
	The proof is by contradiction. Suppose that the evader's position $\e(t)$
  	violates the condition of the theorem at some time $t$. Then, there exist 
	a hyperplane $\bar H$ through $\e(t)$ that separates a subset $\C$ of the
	pursuers from the rest, which we call $\bar{\C}$, such that $|\C| < k$, 
	and therefore $|\bar{\C}| \geq n-k+1$.
	In this case, the evader can escape by moving according to the following strategy:

        \begin{quote}
        {\sl Move with maximum speed in the direction normal to $\bar H$
        towards the side containing $\C$.}
        \end{quote}

	We observe that $\e(t) \not\in \intHull{1} (\bar{\C})$---this is
	because $\e(t)$ lies on $\bar H$, the hyperplane defining the half-space that
	contains $\bar {\C}$. By Lemma~\ref{lem:conv}, therefore, none of the
	$n-k+1$ pursuers in $\bar{\C}$ can catch the evader when the evader
	uses the above-mentioned strategy. Therefore, only the (at most) $k-1$ 
	pursuers in $\C$ can reach the evader, and the $k$-capture of evader 
	is not possible.  This completes the proof.
\end{proof}

The necessary condition asserts that if there is ever a time when the
evader is outside the $k$-Hull of the pursuers, then it has an escape
strategy. The main result of our paper, presented in the following
section, shows that this necessary condition is also sufficient. In
particular, if the evader lies in the pursuers' $k$-Hull at the
initial time instant $t=0$, then the pursuers are able to $k$-capture
it.  (Clearly, if evader is not inside the $k$-Hull initially, then it
can escape unless it plays sub-optimally and move inside the pursuers'
$k$-Hull at a later time, allowing them to capture it.)

\section{Proof of Sufficiency}\label{sec:suff}

In this section, we prove our main result, which is to show that the
necessary condition of Theorem~\ref{thm:necessary} is also sufficient.
The proof, which is constructive, outlines a strategy
for the pursuers and derives an upper bound on the time needed for
the capture. Our analysis exploits properties of the pursuers' $k$-Hull, 
and so we begin with some geometric preliminaries.

\medskip

\subsection{Geometric Preliminaries and an Orientation-Preserving
  Strategy}

\medskip

In general, the orientations of the pursuers with respect to the
evader will change once the pursuit begins. We will show, however,
that pursuers can coordinate their moves to preserve their individual
directions relative to the evader's location. Such a strategy will
allow us to conclude that if the evader is in the $k$-Hull of the
pursuers at the initial instant, then it will remain in the $k$-Hull
at all subsequent instants.

\medskip

Let us call a pursuers' strategy \emph{orientation-preserving} if the
orientations of the vectors $\p_i-\e$ are preserved throughout the
pursuit. We will prove that there is an orientation-preserving
$k$-capture strategy for the pursuers.  But first, we establish a key
geometric lemma about such an strategy.

\medskip

Let $\mathbf{u}_e(t)$ denote the evader's move at time $t$, where the vector 
$\mathbf{u}_e$ is a point on the $m$-dimensional sphere $\mathbb{S}$. 
Let $\theta_i(\mathbf{u}_e)$ denote the (smaller) angle between vectors 
$\p_i(t)-\e(t)$ and $\mathbf{u}_e(t)$. Define,
\[
g(\mathbf{u}_e):= \kmax \{\cos\theta_1(\mathbf{u}_e),\ldots,\cos\theta_n(\mathbf{u}_e)\},
\]
where $\kmax$ refers to the $k$-th maximum of the $n$ quantities. 

The following result states that as long as the pursuers follow an
orientation-preserving strategy, one can always find $k$ favorable
pursuers at each instant of time, for whom the $k$ respective $\theta$'s are all
less than a number which remains invariant at all times and which is
strictly less than $\pi/2$.

\begin{lemma} \label{lem:initial_k} Suppose that the evader lies
  inside the $k$-Hull of the pursuers' initial locations, and the
  pursuers follow an orientation-preserving strategy throughout the
  pursuit. Then, the following facts hold at all times:
\begin{itemize}
	\item There exists a $\subscr{\beta}{max} \:<\: \pi/2$, such
          that at every instant of time,
\begin{equation}\label{eq:beta_k}
\subscr{\beta}{max}:= \arccos\left(\min_{\mathbf{u}_e \in \mathbb{S}}g(\mathbf{u}_e)\right).
\end{equation}
	\item After any move by the evader at time $t+1$, there exist at least $k$
	         pursuers $\p_{i_1},\dots,\p_{i_k}$ such that
  		$\theta_{i_j} \;\leq\; \subscr{\beta}{max}$,
		for all $j \in \{i_1,\ldots,i_k\}$.
\end{itemize}
\end{lemma}
\begin{proof}
Since $\e(0) \in \intHull{k} (\p_1(0),\dots,\p_n(0))$, an orientation
preserving strategy will ensure that $\e(t) \in \intHull{k}
(\p_1(t), \dots,\p_n(t))$, for all time instants $t$. Thus, for any
$t$, the quantity $g(\mathbf{u}_e)$ is identically defined as for $t=0$, and
therefore for the first claim, it suffices to show the existence of a
$\subscr{\beta}{max} < \pi/2$ at time $t=0$, which satisfies Eq.~\eqref{eq:beta_k}.

To see this, we can write $g(\mathbf{u}_e)$ at time $t=0$ as
 \[
g(\mathbf{u}_e) = \kmax \left \{\frac{(\p_1(0)-\e(0)) \cdot
\mathbf{u}_e}{\norm{\p_1(0)-\e(0)}},\ldots,\frac{(\p_n(0)-\e(0)) \cdot
\mathbf{u}_e}{\norm{\p_n(0)-\e(0)}}\right\},
\]
and therefore $g(\cdot)$ is a continuous function of $\mathbf{u}_e$. Since
$\mathbf{u}_e \in \mathbb{S}$, which is a compact set, $g(\cdot)$ attains a
minimum for some $\mathbf{u}_e^*$ in $\mathbb{S}$. It now remains to show
that the minimum value $g(\mathbf{u}_e^*)>0$. Now, for every choice of $\mathbf{u}_e$, we must
have at least $k$ pursuers $\p_{i_1},\dots,\p_{i_k}$ such that
$\theta_{j} \;<\; \pi/2$, for all $j \in \{i_1,\ldots,i_k\}$. If
this were not the case for some $\mathbf{\bar u}_e$, then the hyperplane
perpendicular to $\mathbf{\bar u}_e$ through $\e$ would separate
$k-1$ pursuers from the remaining, implying that $\e\notin \intHull{k}
(\p_1,\dots,\p_n)$. Thus, for every $\mathbf{u}_e \in \mathbb{S}$,
$g(\mathbf{u}_e)>0$ and in particular, $g(\mathbf{u}_e^*) > 0$. Thus,
$\subscr{\beta}{max} = \arccos(g(\mathbf{u}_e^*)) < \pi/2$. Thus, the first claim is established.

\medskip

The second claim follows from the fact that there always exist some
$k$ pursuers $\p_{i_1},\dots,\p_{i_k}$ such that
$\theta_{j} \;<\; \pi/2$, for all $j \in \{i_1,\ldots,i_k\}$, and
since $\cos\subscr{\beta}{max} \leq \cos\theta_{j}$.
\end{proof}

\begin{remark}[General Position]
  Throughout this section, we assume that no two 
  pursuers are collinear with the evader, which implies that the
  vectors $\p_i(0)-\e(0)$ all have distinct orientations at $t=0$, for all $1\leq i \leq n$.
We could easily ensure this condition by an initial move by the pursuers, as follows.
Suppose $\angle \p_i(0)\e(0)\p_{j}(0) = 0$, for some $i,j$, where the
notation $\angle \p\,\x\,\q$ denotes the (smaller) angle between vectors
$\p-\x$ and $\q-\x$. Let the evader's
initial move is from position $\e(0)$ to $\e(1)$. Then, all the pursuers except $\p_i$ 
move parallel to $\e(1)-\e(0)$ with step size $\norm{\e(1)-\e(0)}$. The pursuer $\p_i$ 
also moves with step size $\norm{\e(1)-\e(0)}$ but in a direction making a sufficiently 
small but positive angle $\alpha$ with $\e(1)-\e(0)$. Since $\intHull{k}(\p_1,\ldots,\p_n)$
is an open 
set and a continuous function of the pursuer locations, there exists a sufficiently 
small but positive angle $\alpha$ so that $\e(1)$ still lies inside $\intHull{k}(\p_1,\ldots,\p_n)$ at time $t=1$. If there are multiple collinearities, then the same
strategy can be used to break all of them while preserving the invariant that
the evader lies inside the $k$-Hull. 
\end{remark}

We are now ready to describe our $k$-capture strategy and prove its correctness.

\subsection{A Strategy for {\large $k$}-Capture}\label{sec:strategy}

\medskip

One simple-minded strategy for capture is to let each pursuer
maximally advance towards the evader's new position at each
move. Because the evader lies in $\intHull{k}$, this strategy reduces
at least one pursuer's distance to $\e$.  But it does not ensure that
$k$ pursuers reach the evader simultaneously and so cannot guarantee
$k$-capture. Instead, we let only those pursuers that are \emph{not
  the closest} to the evader execute this kind of move, while those
closest to the evader carry out a \emph{parallel} move that maintains
their distance and angle to the evader.  We call this the
\emph{advance move}.  More specifically, the pursuers who are closest
to the evader move to maintain their distance and angle to the evader,
while the remaining pursuers advance towards the evader.
Unfortunately, while this strategy keeps the pursuers safe, it also
keeps them away from the evader, and in the worst-case all the
pursuers may become equidistant to the evader and then stagnate.  We,
therefore, introduce a second move, called the \emph{cone move}, which
ensures that the distance of the closest pursuers itself decreases but
in such a way that at least $k$ pursuers remain closest to the evader.

\medskip

The following algorithm describes at a pseudo-code level the overall
strategy.  The terms $\subscr{\PP}{closest}$ and $\Cone$,
respectively, denote the set of closest pursuers and a Cone region,
and are defined precisely following the algorithm.

\begin{algorithm}[htbp]
  \KwAssumes{$\e(0)$ satisfies the $k$-Hull necessary condition.} %
  \textbf{For} each $t = 1,2, \ldots$ and for each $j\in\{1,\dots,n\}$,\\
  Determine $\subscr{d}{min}(t) = \norm{\p-\e(t)}$, where $\p \in \subscr{\PP}{closest}(t)$\\
  \eIf{\textup{$\p_j$ is among $k$ pursuers
      $\p_{i_1},\ldots,\p_{i_k}$ that are in
      $\subscr{\PP}{closest}(t)$ and in $\Cone(k,t)$}}
  {
   $\p_j$ uses Cone move corresponding to $\p_{i_1},\ldots,\p_{i_k}$
}{
   $\p_j$ uses Advance move with parameter $\subscr{d}{min}(t)$
}
  \textbf{end for} \\

\caption{\bf $k$-Capture}
\label{algo:strategy_k}
\end{algorithm}

In the following, we give a precise definition of the Advance move and
the Cone move.  Informally, the \emph{Advance} move is used by a
pursuer to reduce its distance from the evader if it is sufficiently
far from the evader.  The \emph{Cone} move is used by a pursuer
\emph{together} with at least $k-1$ other pursuers, if all of them are among the
closest to the evader, and if the evader has made a move which is favorable
for those pursuers.  When both the moves are possible for a pursuer, the
Cone move has the priority.

\begin{definition}[Advance Move]\label{def:planes}
	Suppose the evader moves from $\e(t)$ to $\e(t+1)$. Then, given a parameter 
	$d\geq 0$, the \emph{Advance move} of a pursuer $\p_j$ is the following:
	\begin{itemize}	\advance\itemsep by -4pt
	\item Draw a line through $\e(t+1)$ parallel to the vector $\p_j(t)-\e(t)$. 
	\item Move to the position $\p_j(t+1)$ on this line for which
	$\abs{d-\norm{\e(t+1)-\p_j(t+1)}}$ is minimized. 
	\end{itemize}
\end{definition}

\begin{figure}[htbp]
\centering 
\includegraphics[width=\columnwidth]{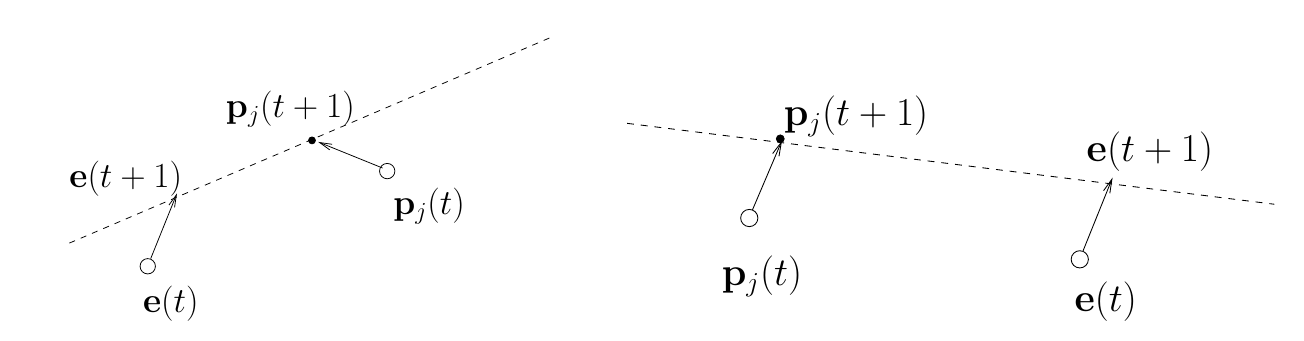}
\caption{Illustrating the Advance move.}
\label{fig:planes}
\end{figure}

Fig.~\ref{fig:planes} illustrates the Advance move.  Our modification
to the original Planes algorithm~\cite{SK-CVR:05} is the inclusion of
the parameter $d$.  This parameter keeps a pursuer $\p_j$ from moving
straight towards $\e$ \emph{if} that pursuer is the closest one to the
evader.  Therefore, with the parameter setting $d =
\norm{\e(t)-\p_j(t)}$, the Advance move forces the pursuer $\p_j$ to
move \emph{parallel} to the evader, and with exactly the same step
size as the evader. This is shown in the right subfigure. The left
subfigure shows a generic application of the Advance move.

We now describe the Cone move, which is used by $k$ or more pursuers
when they are among the closest pursuers to the evader, and when they
are located inside a Cone region, which we define next.  We show later
(cf.~Lemma~\ref{lem:strategy_k}) that after a finite time, there will
be at least $k$ closest pursuers, so the following discussion focuses
on such pursuers.

\medskip

Let $\subscr{\PP}{closest}(t)$ denote the set of pursuers that are closest
to the evader $\e(t)$ at time $t$. That is,

\[
\subscr{\PP}{closest}(t):= \{\p_i(t) \,: \, i \in \argmin_{1,\dots,n} \norm{\p_i(t)-\e(t)} \}.
\]

\begin{definition}[Cone]
	The closed positive cone formed with vertex at $\e(t)$, the
        axis along $\e(t+1)-\e(t)$ (i.e., along $\mathbf{u}_e(t)$),
        and with half angle equal to $\subscr{\beta}{max}$ is
	called the $\Cone(k,t)$.
\end{definition}

\begin{definition}[Cone Move]\label{def:cone_k}
  	If some $k$ pursuers $\p_{i_1}(t),\ldots,\\ \p_{i_k}(t)$ are in 
        $\subscr{\PP}{closest}(t)$ and also in $\Cone(k,t)$, then the Cone move for
	$\p_{i_1},\ldots,\p_{i_k}$ is defined as follows:
	\begin{itemize}\advance\itemsep by -4pt
	\item draw a line $l_j$ through $\e(t+1)$, parallel to
          $\p_j(t)-\e(t)$, for all $j \in \{i_1,\ldots, i_k\}$.
	\item $\p_j(t+1)$ is the point on line $l_j$ that minimizes
          $\norm{\p_j(t+1)-\e(t+1)}$ subject to the constraint that \\
          $\norm{\p_{i_1}(t+1)-\e(t+1)}=\ldots=\norm{\p_{i_k}(t+1)-\e(t+1)}$.
	\end{itemize}
\end{definition}

\begin{figure}[htbp]
\centering 
\includegraphics[width=0.6\columnwidth]{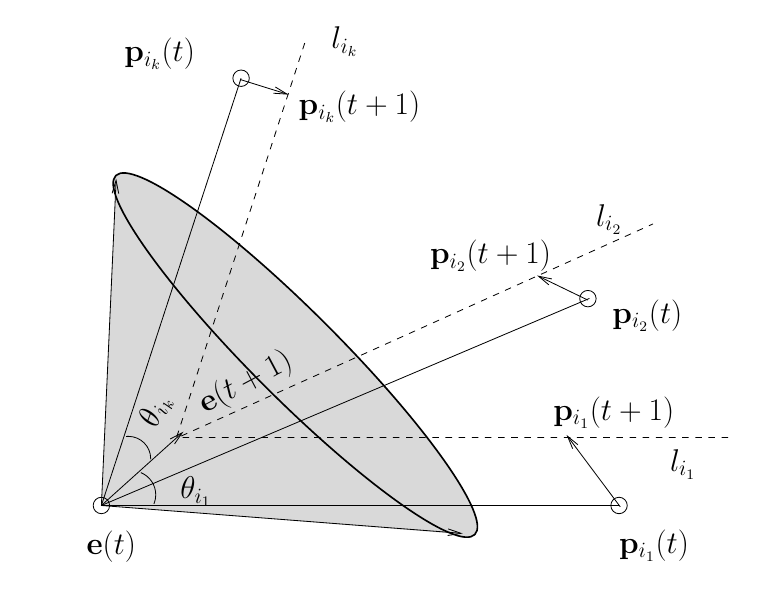}
\caption{Illustrating the Cone Move for $k$ pursuers. The shaded region is
  $\Cone(k,t)$. If the evader moves into $\Cone(k,t)$, then the
  pursuers $\p_{i_1},\dots,\p_{i_k}$ move so that their distances to the
  evader decrease by the same finite amount.}
\label{fig:closein_k}
\end{figure}

Fig.~\ref{fig:closein_k} offers a geometric interpretation of the Cone move.
The key intuition behind the Cone move is that the $k$ closest pursuers in $\Cone(k,t)$ can reduce their distance to the evader,
and remain in $\subscr{\PP}{closest}(t+1)$, using the Cone move.
We can give a more precise expression for the new locations in a Cone move,
as follows.

\medskip

Assume without loss of generality that $\theta_{i_1}\geq \theta_{j},
\, \forall j \in \{i_1,\ldots,i_{k}\}$ . Then, choose 
$\p_{i_1}(t+1)$ satisfying the following conditions:

\begin{itemize}	\advance\itemsep by -4pt
\item $\norm{\p_{i_1}(t+1)-\p_{i_1}(t)} = 1$, and
\item $\angle \p_{i_1}(t+1)\p_{i_1}(t)\e(t) = \arcsin(u_e(t)\sin\theta_{i_1})$,
where $u_e(t) = \norm{\e(t+1)-\e(t)}$.
\end{itemize}
The positions $\p_j(t+1)$, for all $j \in\{i_2,\ldots, i_k\}$, are chosen to
satisfy the following conditions:
\begin{itemize}	\advance\itemsep by -4pt
\item$\norm{\p_{j}(t+1)-\p_{j}(t)}^2 = 1 + u_e^2(\cos\theta_{j}-\cos\theta_i)^2 + 
  2u_e(\cos\theta_{j}-\cos\theta_i)\sqrt{1-u_e^2\sin^2\theta_i}$.
\item $\angle\p_{j}(t+1)\p_{j}(t)\e(t) = \arcsin(u_e(t)\sin\theta_{j})$.
\end{itemize}

\medskip

\subsection{Proof of {\large $k$}-Capture Sufficiency}\label{sec:strategy_k}

\medskip

In the rest of this section, we prove that
Algorithm~\ref{algo:strategy_k} succeeds. 
We begin with the observation that $k$-Capture is
orientation-preserving, since throughout the algorithm, the direction
of the vectors $\p_j-\e$ remains invariant for each $j$.

\begin{proposition}[Orientation Preserving]\label{prop:invariants_k}
	The Algorithm $k$-Capture is orientation-preserving.
\end{proposition}
Our proof of $k$-Capture depends on three technical lemmas showing, respectively, 
that some $k$ pursuers become closest to the evader, that every 
cone move reduces the minimum distance by a finite amount, and that irrespective of 
the evader's strategy, the minimum distance decreases by a finite amount.
Throughout the following discussion, it is assumed that the pursuers all follow 
the Algorithm $k$-Capture.

\medskip

The bound on the capture time depends on 
$\subscr{d}{min}(0):=\min_{i=1}^n \norm{\p_i(0)-\e(0)}$ and
$\subscr{d}{max}(0):=\max_{i=1}^n \norm{\p_i(0)-\e(0)}$,
which are the minimum and the maximum distance between a pursuer
and the evader at the initial time $t=0$.
The following lemma proves the closest pursuers property.

\begin{lemma}[$\subscr{\PP}{closest}$ cardinality]\label{lem:strategy_k}
 After a finite time upper
  bounded by $n(1+\subscr{d}{max}/\subscr{\beta}{max})$, some $k$
  pursuers are in the set $\subscr{\PP}{closest}$.
\end{lemma}
\begin{proof}
  From statement 2 of
  Lemma~\ref{lem:initial_k}, at every instant of time and for any move
  of the evader, there exists some $k$ pursuers
  $\p_{i_1},\dots,\p_{i_k}$ such that for all $j \in \{ i_1,\ldots, i_k\}$,
  $\theta_j \leq \subscr{\beta}{max}$. If all of
  these $k$ pursuers are in $\subscr{\PP}{closest}(t)$, then this result
  stands proved. Otherwise, for every $t$, there exists some pursuer
 (say $\p_j(t)$) out of the $k$ pursuers, which is not in
  $\subscr{\PP}{closest}(t)$, and is such that
  $\theta_j \leq \subscr{\beta}{max}$. So at time
  $t+1$, the Advance move by $\p_j$ will ensure that either
  $\norm{\p_j(t+1)-\e(t+1)}\leq
  \norm{\p_j(t)-\e(t)}-\cos\subscr{\beta}{max}$ or
  $\p_j(t+1)\in\subscr{\PP}{closest}(t+1)$.

\medskip

Thus, in the worst case, after at most
$n(1+\subscr{d}{max}/\cos\subscr{\beta}{max})$ time instants, some $k$
pursuers must be $\subscr{\PP}{closest}$.
\end{proof}

Let $\subscr{d}{min}(t)$ be the distance of the closest pursuer from the 
evader at time $t$. Once $k$ pursuers are in $\subscr{\PP}{closest}$, 
the following lemma establishes a lower bound on the decrease of $\subscr{d}{min}$
assuming that a Cone move occurs, which is favorable for the pursuers. 

\begin{lemma} \label{lem:simcap_k} Let $\p_{i_1}, \dots, \p_{i_k} \in
  \subscr{\PP}{closest}$ be $k$ pursuers closest to the evader at time
  $t$. If these pursuers' next move is a Cone move, then after the
  pursuers' move, we have
\[
\subscr{d}{min}(t+1) \leq \subscr{d}{min}(t) - \cos\subscr{\beta}{max}.
\]
\end{lemma}

\begin{proof}
Let $\theta_j$ be the largest among the angles
$\theta_{i_1},\dots,\theta_{i_k}$. Using the new locations of the pursuers in the Cone move, we obtain,
\begin{align*}
\subscr{d}{min}(t)-\subscr{d}{min}(t+1) &= u_e\cos\theta_j +
1\cdot \cos \angle \p_j(t+1)\p_j(t)\e(t) \\
&= u_e\cos\theta_j+\sqrt{1-u_e^2\sin^2\theta_j}\\
&\geq \cos\theta_j \geq \cos\subscr{\beta}{max},
\end{align*}
since $\theta_j\leq
\subscr{\beta}{max}$ from the definition of the Cone region. The lemma
follows.
\end{proof}

Finally, the next lemma derives a lower bound on the decrease of $\subscr{d}{min}$ 
for the worst-case evader move, while the pursuers follow the strategy of 
Algorithm $k$-Capture.

\begin{lemma}   \label{lem:dist_k}
  If some $k$ pursuers become closest to the
  evader at some time $t$, then the following holds:
\begin{itemize}	\advance\itemsep by -4pt
\item after every subsequent pursuer move, some $k$ pursuers are in $\subscr{\PP}{closest}$, and
	\item after at most
          $n(1+\subscr{d}{max}/\cos\subscr{\beta}{max})$ pursuer moves,
  	$\subscr{d}{min}$ decreases by at least $\cos\subscr{\beta}{max}$.
	\end{itemize}
\end{lemma}
\begin{proof}
  Let $A$ and $B$ be two groups of pursuers in $\subscr{\PP}{closest}$
  at time $t$, of which group $A$ comprises of some $k$ pursuers. If
  all pursuers of group $A$ are in the Cone region at time $t$, then
  group $A$ will make a Cone move which ensures that all pursuers in
  $A$ are in $\subscr{\PP}{closest}$ at time $t+1$. Thus, the first
  claim trivially holds. Otherwise, all pursuers in $A$ move parallel
  to the evader at time $t+1$. Now, if group $B$ does not contain $k$
  pursuers, then at time $t+1$, all pursuers in group $B$ are forced
  to move parallel to the evader, since they do not satisfy the
  criterion to make a Cone move. Thus, the pursuers in group $A$
  satisfy the first claim at time $t+1$. Finally, if group $B$
  contains some $k$ pursuers and are in the Cone region at time $t$,
  then these $k$ pursuers make a Cone move and
  satisfy the first claim at time $t+1$. Thus, the first claim holds
  at all times.

\medskip

Now, let us consider the second claim.
From Proposition~\ref{prop:invariants_k} and statement 2 of
Lemma~\ref{lem:initial_k}, at every instant of time and for any move
of the evader, there exists some $k$ pursuers
$\p_{i_1},\dots,\p_{i_k}$ such that $\theta_j(t) \leq
\subscr{\beta}{max}$, for all $j\in\{i_1,\dots, i_k\}$. We need to consider two cases:

\begin{itemize}
\item {[All of $\p_{i_1}(t),\dots,\p_{i_k}(t)$ are in
    $\subscr{\PP}{closest}(t)$:]}\\
  In this case, the claim follows from Lemma~\ref{lem:simcap_k} because
  all of these pursuers lie in $\Cone(k,t)$.  

\item {[At least one of out of the $k$ pursuers, say $\p_j(t)$ is not in $\subscr{\PP}{closest}(t)$:]}\\
  Without loss of generality, assume that $\p_j (t) \not\in 
  \subscr{\PP}{closest}(t)$.  Then, at time $t+1$, the Advance move by
  $\p_j$ will ensure that either $\norm{\p_j(t+1)-\e(t+1)}\leq
  \norm{\p_j(t)-\e(t)}-\cos\subscr{\beta}{max}$ or
  $\p_j(t+1)\in\subscr{\PP}{closest}(t+1)$.
  Thus, in the worst-case, it requires at most
  $n(1+\subscr{d}{max}/\cos\subscr{\beta}{max})$ moves before all $n$
  pursuers are in $\subscr{\PP}{closest}$. Then, the next pursuer move
  is necessarily a Cone move, because for any choice of the evader
  move, there exists some $k$ pursuers which are now equidistant from
  the evader, which lie in the Cone region.  By
  Lemma~\ref{lem:simcap_k}, the distance of the $k$ closest pursuers
  from the evader strictly decreases by at least
  $\cos\subscr{\beta}{max}$.
\end{itemize}
This completes the proof of the lemma.
\end{proof}

We can now state our main theorem on $k$-Capture.

\begin{theorem}\label{thm:suff_k}
	If the evader lies in the interior of the pursuers' $k$-Hull at $t=0$,
	i.e., $\e(0) \in \intHull{k} (\p_1 (0), \ldots,  \p_n (0))$, then
	it can be $k$-Captured in at most 
	$n(1+ \subscr{d}{max}/\cos\subscr{\beta}{max})^2$ moves.
\end{theorem}
\begin{proof}
  By Lemma~\ref{lem:strategy_k}, after at most
\[
  n(1+\subscr{d}{max}/\subscr{\beta}{max})
\]
moves, some $k$ pursuers are in $\subscr{\PP}{closest}$. Thereafter,
Lemma~\ref{lem:dist_k} ensures that the distance of some $k$
closest pursuers to
the evader decreases by at least $\cos\subscr{\beta}{max}$ after every 
$n(1+ \subscr{d}{max}/ \cos\subscr{\beta}{max})$ moves.
Since capture is defined after the pursuers' move, after at most $n(1+
\subscr{d}{max}/\cos\subscr{\beta}{max})\subscr{d}{max}/\cos\subscr{\beta}{max}$ pursuer
moves, we obtain $\subscr{d}{min} = 0$, that is, the evader and some $k$ 
pursuers are coincident, which satisfies the conditions
of $k$-capture.
An upper bound on the time taken for the $k$-capture of the evader 
follows by summing the bounds of Lemma~\ref{lem:strategy_k} and Lemma~\ref{lem:dist_k}. 
This completes the proof of the theorem.
\end{proof}

\begin{remark}[Lower bound on Capture time]
  A lower bound on the time taken to capture is
  $\subscr{d}{max}/\cos\subscr{\beta}{max}$. To see this, consider the
  following initial condition and evader strategy. The evader's
  strategy is to move along a fixed vector $\mathbf{u}_e$ with unit
  step. Let $\p_1, \dots, \p_k$ be furthest from the evader initially,
  and be located on the boundary of the resulting $\Cone(k,0)$. The
  rest of the pursuers are located outside $\Cone(k,0)$. This evader
  strategy and the initial pursuer locations ensure that the evader
  is captured after a time of at least
  $\subscr{d}{max}/\cos\subscr{\beta}{max}$, independent of the
  pursuers' strategy.
\end{remark}

\section{Bounded Environments}\label{sec:compact}

In this section, we show a simple strategy for $k$-capture that always succeeds
in a compact and convex subset of a Euclidean space. If every pursuer were to 
use an established strategy
by Sgall~\cite{JS:01} independently of the other pursuers, at each
instant of time, then the distance between each pursuer and the evader
would decrease to zero, but at different instants in time. Although
this approach does not guarantee $k$-capture in general, it suggests
that intuitively, it should be possible to coordinate the moves of
each pursuer to achieve $k$-capture from any set of initial locations
in the environment. Therefore, in contrast with the previous sections
wherein there existed a necessary condition for $k$-capture, we will
now directly present a strategy which requires $k$ pursuers, and which
achieves $k$-capture of the evader in at most $O(D^2)$ time steps,
where $D$ is the diameter of the environment.

\medskip

Our strategy comprises of two phases. The first phase is an
\emph{initializing} move, which gets the pursuers in a favorable
formation so that they can apply the steps in the second phase. In
particular, the initializing move will show that it is possible to
achieve a configuration of the pursuers and the evader such that $k-1$
pursuers are located between a \emph{lead} pursuer and the evader.

\medskip

The second phase will mimic Sgall's strategy~\cite{JS:01} for the
lead pursuer, while the other $k-1$ pursuers will maintain an
invariant of being located between the lead pursuer and the
evader at all times. The initial locations of the pursuers being
sufficiently close to each other ensures that the evader gets captured
if it moves to the location of any pursuer. We show that this phase
terminates into the evader being $k$-captured.

Let us begin with the Initializing move.

\subsection{Initializing Move}

\medskip

In this phase, the pursuers first group themselves such that they are
located inside a sphere of radius equal to half. This essentially
means that every pursuer can reach the location of any other pursuer,
in one time step.

\medskip

Now, consider a closed sphere $O$ of radius half which contains the
pursuers at time $t = 0$. Let $\ell$ denote the intersection of the
sphere $O$ with line joining the evader's location at time $t=1$ to
the center of $O$. Now, independent of the location $\e(1)$, it is
always possible to find $k$ distinct locations $\p_1(1),\dots,\p_k(1)$
each contained in $\ell$, such that $\p_1(1),\dots,\p_k(1)$ are
collinear with $\e(1)$ and $\p_2(1),\dots,\p_k(1)$ lie between
$\p_1(1)$ and $\e(1)$. Figure~\ref{fig:initial} shows an illustration
of this move.

\begin{figure}[h]
\centering 
\includegraphics[width=0.3\columnwidth]{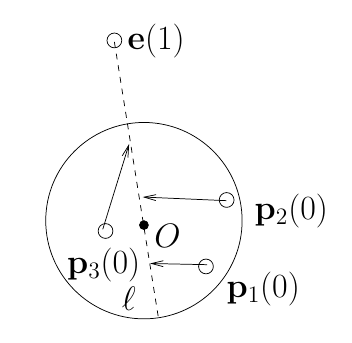}
\caption{Illustrating the initializing move. It is always possible to
  ensure that the pursuers are collinear with the evader and within a
  unit distance of each other. In this figure, the circle centered at
  $O$, has radius equal to half unit.}
\label{fig:initial}
\end{figure}

This terminates the initializing move, and we are now ready to present
the $k$-capture strategy.

\medskip

\subsection{An algorithm for $k$-Capture}

\medskip

At each time instant $t$, $\p_1$ makes the Sgall move, described as
below.

\begin{enumerate}
\item Join $\e(t-1)$ and $\p_1(t-1)$ and extend it beyond $\p_1(t-1)$
  to intersect the environment at $C$. 
\item Move to the point closest to $\e(t)$ and on the line joining
  $\e(t)$ and $C$.
\end{enumerate}

All other pursuers pick distinct points between $\p_1(t)$ and
$\e(t)$. This strategy is illustrated in Figure~\ref{fig:sgall}, and
is summarized in Algorithm~\ref{algo:compact}.  

\begin{figure}[htbp]
\centering 
\includegraphics[width=0.3\columnwidth]{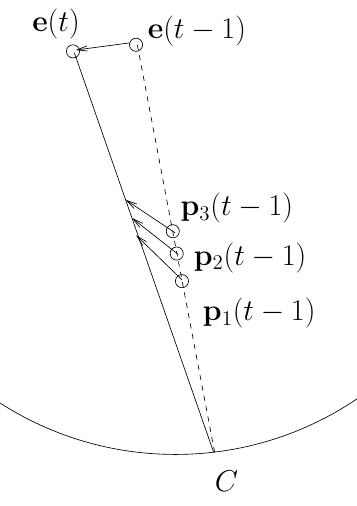}
\caption{Illustrating a move of Algorithm~\ref{algo:compact}. Pursuer
  $\p_1$ follows the Sgall move, while all the others pick distinct points
  between $\p_1$ and $\e$ to move to.}
\label{fig:sgall}
\end{figure}

\begin{algorithm}[h]
 \KwAssumes{The players are in a configuration resulting from the
    Initializing move.} %
  \textbf{For} each $t = 1,2, \ldots$,\\
\quad $\p_1$ makes the \emph{Sgall Move}.\\
\quad \textbf{For} each $j\in \{2,\dots,k\}$,\\
\qquad $\p_j$ moves to the furthest point from $\p_1$ between $\p_1(t)$ and
$\e(t)$, and on the line
joining $\p_1(t)$ and $\e(t)$. \\
\quad \textbf{end for} \\
 \textbf{end for} \\
\caption{\bf Sgall-like strategy}
\label{algo:compact}
\end{algorithm}

Thus, we obtain the following result.

\begin{proposition}
With the initializing move and subsequently
Algorithm~\ref{algo:compact}, the pursuers $k$-capture the evader
in $O(D^2)$ number of time steps, where $D$ is the diameter of the
compact environment.
\end{proposition}
\begin{proof}
  Since $\p_1$ uses the Sgall move throughout the pursuit,
  $\norm{\e-\p_1}$ becomes equal to zero in at most $O(D^2)$ number of
  time steps. Each pursuer move in step 4 of the algorithm exists
  since the environment is convex. Thus, the remaining $k-1$ pursuers
  ensure that the evader is $k$-captured when $\p_1$ becomes
  coincident with $\e$.
\end{proof}

\section{Closing Remarks}\label{sec:conclusions}

In this paper, we introduced a new variant of the classical
pursuit-evasion problem in an $m$-dimensional Euclidean space, which requires
multiple pursuers to simultaneously reach the evader for capture. We
showed that, for $k$-capture to occur, the evader must lie inside the
$k$-Hull, in a pleasing generalization of the convex hull rule for the
single pursuer capture. The main result of the paper was to show that
this simple necessary condition is also sufficient. The proof of this
sufficiency required a new pursuit strategy, requiring both an Advance
move, which is a modified version of a known Planes algorithm and a
new type of Cone move, which requires a careful coordination among the
pursuers. For a version of this problem played in a compact and convex
environment, we showed that $k$-capture is always possible.

\medskip

Our work suggests a number of intriguing problems for future research.
Interesting directions include improving the upper bound on the time
taken to capture the evader and addressing versions of this problem in
general environments, with obstacles.

\end{document}